\documentclass{sapm}

\usepackage{amsmath,amsthm,amssymb,mathrsfs}
\usepackage{graphicx}
\usepackage{color}

\newcommand{\domain}{G}
\newcommand{\Log}{\mathrm{Log}}
\newcommand{\R}{\mathrm{Re}}
\newcommand{\I}{\mathrm{Im}}

\newcommand{\cc}{\mathrm{c.c.}}
\newcommand{\diag}{\mathrm{diag}}

\def\Xint#1{\mathchoice
{\XXint\displaystyle\textstyle{#1}}%
{\XXint\textstyle\scriptstyle{#1}}%
{\XXint\scriptstyle\scriptscriptstyle{#1}}%
{\XXint\scriptscriptstyle\scriptscriptstyle{#1}}%
\!\int}
\def\XXint#1#2#3{{\setbox0=\hbox{$#1{#2#3}{\int}$ }
\vcenter{\hbox{$#2#3$ }}\kern-.6\wd0}}

\def\dashint{\Xint-}

\jname{}
\jvol{Pages}
\jyear{}
\doi{}

\begin{document}

% Title
\title{Direct Scattering for the Benjamin-Ono Equation with Rational Initial Data}

%Authors, affiliations address.
\author[P. D. Miller and A. N. Wetzel]{Peter D. Miller and Alfredo N. Wetzel\thanks{Address for
correspondence:  Alfredo N. Wetzel, Department of Mathematics, University of Michigan, East Hall, 530 Church St., Ann Arbor, MI  48109, USA; e-mail:  \texttt{wreagan@umich.edu}}}
\affil{University of Michigan}

\maketitle

%Abstract
\begin{abstract}
We compute the scattering data of the Benjamin-Ono equation for arbitrary rational initial conditions with simple poles.
Specifically, we obtain explicit formulas for the Jost solutions and eigenfunctions of the associated spectral problem, yielding an Evans function for the eigenvalues and formulas for the phase constants and reflection coefficient.
\end{abstract}

% Introduction
\section{Introduction}\label{introduc}

The Cauchy problem of the Benjamin-Ono (BO) equation
\begin{equation}\label{BOeqn}
\frac{\partial u}{\partial t} + 2 u \frac{\partial u}{\partial x} + \epsilon \mathcal{H}\left[ \frac{\partial^2 u}{\partial x^2} \right] =0, \quad u(x,0) = u_0(x), \quad x \in \mathbb{R}
\end{equation}
describes the weakly nonlinear evolution of one-dimensional internal gravity waves in a stratified fluid \cite{Benjamin67, DavisAcrivos67, Ono75}. The variable $u$ corresponds to the wave profile, $\epsilon$ is a positive real number that quantifies the effects of dispersion, and the operator $\mathcal{H}$ denotes the Hilbert transform defined by the Cauchy principal value integral
\begin{equation}
\mathcal{H}[u](x,t) := \frac{1}{\pi} \dashint_{-\infty}^{\infty} \frac{u(\xi,t)}{\xi-x}\; d\xi.
\end{equation}

The Cauchy problem \eqref{BOeqn} is solvable, for suitable $u_0$ decaying as $x\to\infty$, by an inverse-scattering transform (IST) associated to a spectral problem (cf.\@ \eqref{Lax1}) in which $u_0$ appears as a potential.  The aim of this paper is to construct explicitly the Jost solutions and ``bound state'' eigenfunctions for this spectral problem with generic  (simple poles) rational $u_0$, and to deduce therefrom the corresponding scattering data.
This work was inspired by a brief note in a paper of Kodama, Ablowitz, and Satsuma  \cite{KodamaAblowitzSatsuma82} on the intermediate long wave (ILW) equation from which BO arises as a limiting case.
In \cite{KodamaAblowitzSatsuma82}, the spectral problem \eqref{Lax1} is analyzed for a ``Lorentzian'' potential ($u_0(x) = 2\nu (1+x^2)^{-1}$ for $\nu\in\mathbb{R}$) and it is shown that the eigenvalues are, under some conditions on $\nu$, the roots of certain Laguerre polynomials.
While some complications and essential difficulties in extending this remarkable calculation
were highlighted by Xu \cite{Xuthesis}, we can now report that a vast generalization is indeed possible.  

It is plausible that a method capable of treating a general class of rational potentials can subsequently be extended to arbitrary potentials by an appropriate density argument.
Moreover, our results allow for the analysis of the scattering data in the zero-dispersion limit ($\epsilon \to 0$) \cite{upcoming-paper2}.

% Background
\section{Background}\label{Backgr}

The BO equation is the compatibility condition for a Lax pair \cite{BockKruskal79,Nakamura79b}, and a corresponding IST for \eqref{BOeqn} was first found by Fokas and Ablowitz \cite{FokasAblowitz83} with key simplifications accounting for reality of the potential given later by Kaup and Matsuno \cite{KaupMatsuno98}.  The first step of the IST solution of \eqref{BOeqn} is the association of $u_0$ with a set of \emph{scattering data} via the spectral problem 
\begin{equation}\label{Lax1}
i \epsilon \frac{d w^{+}}{d x} + \lambda (w^{+} - w^{-}) = -u_0 w^{+},\quad x\in\mathbb{R},
\end{equation}
where 
$\lambda \in \mathbb{C}$ is a spectral parameter and the superscripts $\pm$ denote the boundary values of the function $w(x)$, analytic and bounded on $\mathbb{C} \backslash \mathbb{R}$, from the upper- and lower-half $x$-plane on the real line.

Recall the Cauchy operators $\pm\mathcal{C}^{\pm}$ defined by the singular integrals
\begin{equation}\label{Cpmdef1}
\mathcal{C}^{\pm}[\varphi](x) := \lim_{\delta \downarrow 0} \frac{1}{2 \pi i} \int_{-\infty}^{\infty} \frac{\varphi(y)}{y - (x \pm \delta i)} \, dy. \end{equation}
The Cauchy operators $\pm \mathcal{C}^{\pm}$ are bounded operators on $L^2(\mathbb{R})$ and are complementary (due to the  \emph{Plemelj formula} $\mathcal{C}^{+} - \mathcal{C}^{-} = \mathbb{I}$) orthogonal projections  from $L^2(\mathbb{R})$ onto its Hardy subspaces $\mathbb{H}^{\pm}(\mathbb{R})$. 
For details on singular integrals see \cite{Gakhov66,Muskhelishvili53}. Using the Plemelj formula and Liouville's Theorem,  \eqref{Lax1} is projected into the equation
\begin{equation} \label{w+bounded}
i \epsilon \frac{d w^{+}}{d x} + \lambda \left(w^{+} - w_0 \right) = -\mathcal{C}^{+} \left[u_0 w^{+}\right],
\end{equation}
where $w_0$ is a suitable constant \cite{Xuthesis}.  Thus $w^-$ has been eliminated, and  \eqref{w+bounded} is related to the eigenvalue problem $\mathcal{L}w^+=\lambda w^+$ for the self-adjoint operator $\mathcal{L}:=-i\epsilon\partial_x-\mathcal{C}^+u_0\mathcal{C}^+$ acting on $\mathbb{H}^+(\mathbb{R})$ having essential spectrum $\mathbb{R}^+$ and possibly some negative real point spectrum \cite{ReedSimon4}.

To calculate the scattering data of \eqref{w+bounded} one first obtains the 
\emph{Jost solutions}.
These are certain solutions of \eqref{w+bounded} (also analytic and bounded for $\I\{x\}>0$) that are well-defined for $\lambda>0$; 
in the literature  \cite{FokasAblowitz83, KaupMatsuno98, AblowitzFokasAnderson1983, SantiniAblowitzFokas1984} they are denoted by $w^+=M(x;\lambda)$ and $\overline{N}(x;\lambda)$ (for $w_0=1$) and $w^+=\overline{M}(x;\lambda)$ and $N(x;\lambda)$ (for $w_0=0$).  They may be characterized for fixed $\lambda>0$ via their asymptotic behavior for large real $x$ as follows:
\begin{equation}\label{MNbds}
\begin{split}
M \to 1, \; \overline{M} e^{-i \lambda x/\epsilon} \to 1 \quad &\text{as} \quad x \to -\infty \\
\overline{N} \to 1, \; N e^{-i \lambda x/\epsilon} \to 1 \quad &\text{as} \quad x \to +\infty. 
\end{split}
\end{equation}
The Jost solutions solve (second kind) Fredholm integral equations \cite{FokasAblowitz83}. 
It is then possible to show that $M$, $N$, and $\overline{N}$ satisfy the \emph{scattering relation}
\begin{equation} \label{MNbarrelation1}
M(x;\lambda) - \overline{N}(x;\lambda) = \beta(\lambda)N(x;\lambda),\quad\lambda\in\mathbb{R}^+,
\end{equation}
determining a function $\beta:\mathbb{R}^+\to\mathbb{C}$ (independent of $x$), called the \emph{reflection coefficient}. 

We next introduce the ``bound state'' \emph{eigenfunction} $w^+=\Phi_j(x) \in \mathbb{H}^{+}(\mathbb{R})$ satisfying \eqref{w+bounded} with $w_0=0$ for a given \emph{eigenvalue} $\lambda =\lambda_j<0$ and normalized by the condition
\begin{equation}\label{Phij_bds}
x\Phi_j(x) \to 1
\quad \text{as} \quad
|x| \to \infty \quad \text{(uniformly for  $\I\{x\}\ge 0$)},
\end{equation}
or equivalently, as can be shown asymptotically from \eqref{w+bounded}, 
\begin{equation}\label{norm_eigfunc}
\frac{1}{2 \pi i} \int_{-\infty}^{\infty} u_0(x) \Phi_j(x) \, dx = \lambda_j.
\end{equation}

Importantly, while the Jost solutions are initially defined only for $\lambda >0$, $M$ and $\overline{N}$ have analytic extensions into the complex $\lambda$-plane while in general $\overline{M}$ and $N$ do not \cite{FokasAblowitz83}.
Indeed, for each fixed $x$ with $\I\{x\}\ge 0$, $M$ and $\overline{N}$ can be shown to be the boundary values on $\mathbb{R}^{+}$ from the upper- and lower-half $\lambda$-planes respectively of a single function $W$ meromorphic on $\mathbb{C} \setminus \mathbb{R}^{+}$. 
The function $W$ has as its only singularities in $\lambda$ on $\mathbb{C} \setminus \mathbb{R}^{+}$ a discrete set of simple poles located precisely at the eigenvalues $\lambda = \lambda_j<0$; see \cite[Theorem 2.1]{AndersonTaflin1985}.
We will henceforth use the notation $W_{+}$ (resp., $W_-$) to denote the Jost solution $M$ (resp., $\overline{N}$). 
(For the inverse theory, it is then important that $N(x;\lambda)$ may be fully eliminated from \eqref{MNbarrelation1} to yield a \emph{nonlocal jump condition} relating the boundary values $W_\pm$ of $W$ across its branch cut $\mathbb{R}^+$.)
The meromorphic function $W$ then has a Laurent expansion about each eigenvalue $\lambda_j$ of the form:
\begin{equation}\label{LaurentW}
W(x;\lambda) = -i \epsilon \frac{\Phi_j(x)}{\lambda - \lambda_j} + \left( x + \gamma_j \right) \Phi_j(x) + O(\lambda - \lambda_j)
\end{equation}
where $\gamma_j\in\mathbb{C}$ is the \emph{phase constant} associated with the eigenvalue $\lambda_j$.

\begin{definition}[BO scattering data of the potential $u_0$]
\label{SD_BO_problem} 
The scattering data associated with $u_0$ are as follows.
\begin{itemize}
\item Reflection coefficient $\beta(\lambda)$ for $\lambda \in \mathbb{R}^{+}$: defined by \eqref{MNbarrelation1}; see also \eqref{beta_formula}.
\item Eigenvalues $\{\lambda_j<0\}_{j=1}^{N}$: determined from the spectral problem \eqref{w+bounded}.
\item Phase constants $\{\gamma_j\}_{j=1}^{N}$: defined in terms of the eigenvalues and the corresponding normalized eigenfunctions by \eqref{LaurentW}.
\end{itemize}
\end{definition}

% Direct Scattering for Rat IC
\section{Direct scattering for rational potentials}\label{DirScatr}

\subsection{Rational Potentials}\label{subsec:RatIC}

A real bounded potential $u_0$ that is rational with simple poles and that decays as $x\to\pm\infty$ necessarily has the form:
\begin{equation}\label{u0_def_frac}
u_0(x) = \sum_{p=1}^{P} \frac{c_p}{x-z_p}+ \cc,
\end{equation}
where $\{c_p\}_{p=1}^P$ are nonzero complex numbers and the poles $\{z_p\}_{p=1}^{P}$ with $\I\{z_p\}>0$ have distinct real parts increasing with $p$.
We impose in addition the condition  $\sum_{p =1}^{P} \left( c_p + c_p^* \right)=0$,
ensuring $u_0\in L^1(\mathbb{R})$.  

Let $f(x)$ be the particular anti-derivative of $u_0(x)$ given by 
\begin{equation} \label{f_eqn}
f(x) := \sum_{p=1}^{P}  c_p \left( \Log \left( i (x-z_p)\right)  +\frac{\pi i}{2} \right)+ \cc,\quad
f'(x)=u_0(x),
\end{equation}
where $\Log(\cdot)$ is the principal branch ($|\I\{\Log(\cdot)\}|<\pi$).
We denote by $\domain$ the domain of analyticity of $f$ as illustrated in Figure~\ref{fig:fdomain}.
Given $R>0$ sufficiently large, we define subdomains $\domain_p$ of $\domain\cap\{|x|>R\}$, $p=1,\dots,P-1$,
consisting of points lying between the branch cuts emanating from the points $z_{p}$ and $z_{p+1}$.  Then we define $\domain_{0}$ as the half-plane $\R\{x\}<\R\{z_1\}$ and $\domain_P$ as the half-plane $\R\{x\}>\R\{z_P\}$.
\begin{figure}[h!]
\begin{center}
\includegraphics[scale=1]{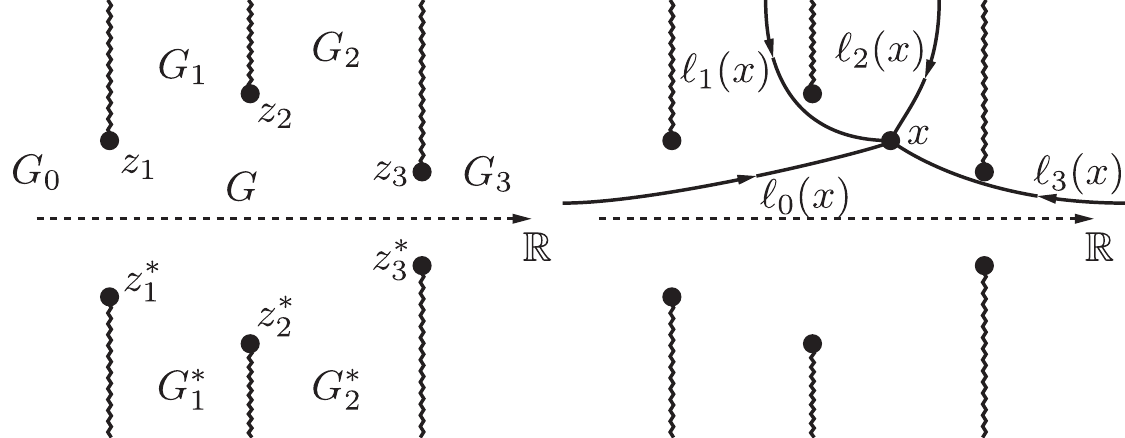}
\caption{
Left: a domain $\domain$ with subdomains $\domain_p$ and $\domain_p^*$ in the $x$-plane for a potential $u_0$ of the form \eqref{u0_def_frac} with $P=3$.  The zig-zagged half-lines denote the (logarithmic) branch cuts of $f(x)$. Right: contours $\ell_m(x)$ in $\domain$ originating at $\infty$ in the subdomains $\domain_m$ and terminating at $x\in\domain$.
} \label{fig:fdomain}
\end{center}
\end{figure}
\begin{lemma}\label{prop_f_to_infty}
For any $m=0,1,\dots,P$, we have
\begin{equation}
\lim_{x\to\infty\in G_m}f(x)=2\pi i\sum_{p=1}^{m} c_p
\quad \text{and} \quad
\lim_{x\to\infty\in G_m^*}f(x)=-2\pi i\sum_{p=1}^{m} c_p^*.
\label{eq:f-asymptotic}
\end{equation}
\end{lemma}

\begin{proof}
This follows immediately from formula \eqref{f_eqn}.
\end{proof}
The advantage of considering potentials $u_0$ of the form \eqref{u0_def_frac}, as first noted in \cite{KodamaAblowitzSatsuma82}, is that upon applying the Plemelj formula to the right-hand side of \eqref{w+bounded}, the expression $\mathcal{C}^-[u_0w^+]$ can be evaluated by residues because $w^+$ is analytic and bounded for $\I\{x\}>0$.  Therefore,
\begin{equation}
i\epsilon\frac{dw^+}{dx}+\lambda w^++u_0w^+=\lambda w_0 + \sum_{p=1}^P\frac{c_pw^+(z_p;\lambda)}{x-z_p}.
\label{w+bounded-rational}
\end{equation}
As the $x$-dependence on the right-hand side is explicit and elementary this first-order equation can be solved in closed-form, at least assuming $w^+(z_p;\lambda)$ is known for all $p$.

\subsection{The Jost solution $M=W_+$ and reflection coefficient $\beta$}\label{subsec:Jost}

Let $\lambda>0$.  To obtain a formula for the Jost solution $W_+(x;\lambda)$ we take $w_0=1$ in \eqref{w+bounded-rational}, apply the boundary condition $W_+(x;\lambda)\to 1$ as $x\to -\infty$, and integrate to obtain the convergent improper integral representation:
\begin{equation}\label{MJost_eqn}
W_{+}(x;\lambda) 
= - \frac{i}{\epsilon} e^{ih(x;\lambda)/\epsilon } \int_{-\infty}^{x}  e^{-ih(z;\lambda)/\epsilon} \left(  \lambda + \sum_{p=1}^{P}   \frac{v_p(\lambda)}{z-z_p}\right) dz,
\end{equation}
where $h(x;\lambda):= \lambda x + f(x)$ and  $v_p(\lambda):=c_pW_+(z_p;\lambda)$ for $p=1,\dots,P$.
Note that by Lemma~\ref{prop_f_to_infty}, $h(x;\lambda)$ is dominated by $\lambda x$ for large $|x|$ provided $\lambda\neq 0$.  Since $\lambda>0$, the integral can therefore be made absolutely convergent by rotation of the contour at infinity into the lower-half $x$-plane.  A similar representation can be obtained for $\overline{N}(x;\lambda)=W_-(x;\lambda)$. 

While $W_+(x;\lambda)$ given by \eqref{MJost_eqn} satisfies the differential equation \eqref{w+bounded-rational} and is analytic for all $x$ near $\mathbb{R}$, extra conditions are required to ensure analyticity for all $x$ with $\I\{x\}>0$.
Indeed, the factor 
\begin{equation}\label{pole_form}
e^{ih(x;\lambda)/\epsilon} = e^{i \lambda x/\epsilon} \prod_{p=1}^{P} \left( i(x-z_p)\right)^{ic_p/\epsilon} \left( i^* (x-z_p^*) \right)^{ic_{p}^*/\epsilon} e^{\pi(c_p^* - c_p)/(2\epsilon)}
\end{equation}
appearing in \eqref{MJost_eqn} may have singularities for $\I\{x\}>0$.
Requiring analyticity of $W_+(x;\lambda)$ imposes constraints on the quantities $\{ v_p(\lambda) \}_{p=1}^{P}$, thus eliminating them entirely from \eqref{MJost_eqn}.

\begin{definition}
\label{defn:ell_m}
$\ell_m(x)$ denotes any member of the equivalence class of contours in $\domain$ originating at $\infty$ in the subdomain $\domain_m$ and ending at the point $x\in \domain$; see Figure~\ref{fig:fdomain}.
\end{definition}

\begin{definition}\label{defn_Uplus}
$U_m^>$ denotes any member of the equivalence class of contours on the Riemann surface of $f$ originating at $-i\infty$ in $\domain_0$ with value $f(\infty)=0$ and ending at $-i\infty$ in $\domain_0$ having encircled the singularities $z_1,\dots,z_m$ each exactly once in the positive sense; see Figure~\ref{fig:Upos_fig}.
\end{definition}
\begin{figure}[h!]
\begin{center}
\includegraphics[scale=1]{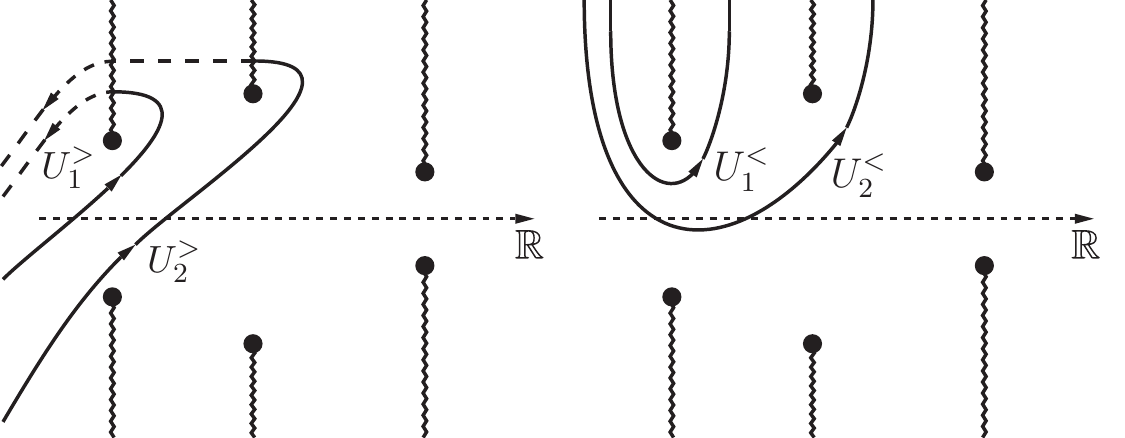}
\caption{
Left: contours $U_{m}^{>}$ for the configuration illustrated in Figure~\ref{fig:fdomain}.
Right: contours $U_{m}^{<}$ for the same configuration.
} \label{fig:Upos_fig}
\end{center}
\end{figure}

\begin{proposition} \label{analyticity1_M}
Let $\lambda>0$.  For each $m=1,\dots, P$, define the function $H_m=H_m(z;\lambda)$, $z\in\domain$, by the equation 
\begin{equation}\label{def_Hj}
 (i(z-z_m))^{-ic_m/\epsilon-1} H_m(z;\lambda)
:= - e^{-ih(z;\lambda)/\epsilon }\left( \lambda +  \sum_{p=1}^{P} \frac{v_p(\lambda)}{z-z_p}\right).
\end{equation}
(Note that $H_m$ is analytic not only on $\domain$, but also in a neighborhood of the point $z_m$.)
Also, let $i c_m = \epsilon(\mu_m +\omega_m)$, where $\mu_m = [\R\{ic_m/\epsilon\}]$
denotes the integer (floor) part of $\R\{ic_m/\epsilon\}$
and the remainder satisfies $0\leq \R\{\omega_m\} < 1$. Then, the function $W_{+}(x;\lambda)$ given by \eqref{MJost_eqn} is analytic in the upper-half $x$-plane if and only if for each $m = 1,\ldots, P$, either
\begin{equation}\label{lzj_cond}
\int_{\ell_{0}(z_m)} (i(z-z_m))^{-i c_m/\epsilon-1} H_m(z;\lambda) \; dz =0
\quad \text{when} \quad
\R\{ic_m/\epsilon\} < 0, 
\end{equation}
\begin{multline} \label{lzj_der_cond}
\int_{\ell_{0}(z_m)} \left( i(z-z_m )\right)^{-\omega_m} \frac{d^{\mu_m+1}}{d z^{\mu_m+1}} H_m(z;\lambda) \; dz =0\\
 \text{when} \quad
\R\{ ic_m/\epsilon\} \geq 0
\quad\text{and}\quad\omega_m \neq 0, 
\end{multline}
or
\begin{equation}\label{lzj_Log_cond}
\mathop{\mathrm{Res}}_{z=z_m}\frac{H_m(z;\lambda)}{(z-z_m)^{\mu_m+1}}
 =0
\quad \text{when} \quad
\R\{ ic_m/\epsilon\} \geq 0
\quad \text{and} \quad \omega_m = 0,
\end{equation}
where $\ell_0(z_m)$ (see Definition~\ref{defn:ell_m}) originates in at $-i\infty$ in $\domain_0$.
\end{proposition}

Proposition~\ref{analyticity1_M} is proved in the Appendix.
Note that the singularity at $z=z_m$ arising from the exponential $e^{-ih(z;\lambda)/\epsilon}$ on the right-hand side of \eqref{def_Hj} is explicitly cancelled from $H_m(z;\lambda)$ by the factor $(i(z-z_m))^{-ic_m/\epsilon-1}$ on the left-hand side.
The essence of the proof is to deform the contour of integration in \eqref{MJost_eqn} to pass through the point $z_m$ when the integrand is integrable there, i.e.,  $\mathrm{Re}\{ic_m/\epsilon\} < 0$. Thus, we deal separately with integrals on $\ell_0(z_m)$ and on the line from $z_m$ to $x$. If the integrand is not integrable at $z_m$, we first integrate by parts repeatedly to achieve integrability, or  invoke the Residue Theorem if the singularity is a pole.

It is useful to rewrite the conditions for analyticity of $W_+(x;\lambda)$ given in Proposition~\ref{analyticity1_M} by eliminating the derivatives of $H_m$ as follows.  

\begin{corollary}\label{Um_neg_system}
Let $\lambda>0$.  The function $W_{+}(x;\lambda)$ defined by \eqref{MJost_eqn} is analytic in the upper-half $x$-plane if and only if $\{v_p(\lambda)\}_{p=1}^P$ satisfy 
\begin{equation}\label{integral_Um_neg}
\int_{C_m^{>}} e^{-ih(z;\lambda)/\epsilon } \left( \sum_{p=1}^{P} \frac{v_p(\lambda)}{z-z_p} + \lambda \right)\,dz = 0,\quad m=1,\dots,P,
\end{equation}
where $C_m^{>} := U_m^{>}$ unless $ic_m/\epsilon$ is a strictly negative integer, in which case $C_m^{>} := \ell_0(z_m)$. (The contours $C_m^>$ extend to $-i\infty$ in $\domain_0$ making the integrals absolutely convergent.)
\end{corollary}

\begin{proof}
Let $L_m$ be a contour on the Riemann surface of $f$ beginning (with value $f(\infty)=0$) and ending in at $-i\infty$ in $\domain_0$, and encircling only the singularity $z_m$ once in the positive sense.
 Now, in the exceptional case when $C_m^{>} = \ell_0(z_m)$, \eqref{integral_Um_neg} is just the condition \eqref{lzj_cond} rewritten using  \eqref{def_Hj}.  In \eqref{lzj_cond} with $\omega_m\neq 0$ as well as in \eqref{lzj_der_cond}, the integrand has an integrable algebraic branch point at $z=z_m$ so up to a nonzero factor of the form $(1-\alpha)/2$ the integrals appearing in \eqref{lzj_cond} and \eqref{lzj_der_cond} can both be rewritten equivalently with the contour $\ell_0(z_m)$ replaced by $L_m$.  In the case of \eqref{lzj_der_cond} we then integrate by parts $\mu_m+1$ times on $L_m$, which produces no boundary terms for $\lambda>0$ due to exponential decay of the integrand at $-i\infty$.  Eliminating $H_m$ by means of \eqref{def_Hj} then gives \eqref{integral_Um_neg} with $C_m^>$ replaced by $L_m$ in these cases.  The final condition of Proposition~\ref{analyticity1_M}, \eqref{lzj_Log_cond}, can of course be written  as an integral about a small loop centered at $z=z_m$, and applying Cauchy's Theorem under the condition $\lambda>0$ allows this loop to be deformed into the contour $L_m$ (in the case of a locally single-valued integrand), so \eqref{lzj_Log_cond} also takes the form \eqref{integral_Um_neg} with $C_m^>$ replaced by $L_m$ after elimination of $H_m(z;\lambda)$. Finally, the replacement of $L_m$ with $U_m^>$ in each case to arrive at \eqref{integral_Um_neg} amounts to multiplication of the system \eqref{integral_Um_neg} by a triangular nonsingular matrix.
\end{proof}

Note that the conditions for analyticity of $W_+(x;\lambda)$ as expressed in Corollary~\ref{Um_neg_system} take the form of a square linear system of equations on the unknowns $\{v_p(\lambda)\}_{p=1}^{P}$ assembled in a vector $\mathbf{v}(\lambda) := (v_1(\lambda),\dots,v_P(\lambda))^\mathsf{T}$:
\begin{equation}\label{lin_sys_Abar}
\mathbf{A}^{>}(\lambda) \mathbf{v}(\lambda) = \mathbf{b}^{>}(\lambda), \quad\lambda>0,
\end{equation}
where  $\mathbf{A}^{>}(\lambda)\in\mathbb{C}^{P\times P}$ and $\mathbf{b}^{>}(\lambda)\in\mathbb{C}^{P}$ have components
\begin{equation}\label{lin_sys_coeff}
A_{mp}^{>}(\lambda) := \int_{C_m^>} \frac{e^{-ih(z;\lambda)/\epsilon }}{z-z_p} \, dz
\; \;\text{and} \;\;
b_m^{>}(\lambda) := -  \lambda \int_{C_m^>} e^{-ih(z;\lambda)/\epsilon } \, dz.
\end{equation}

The reflection coefficient $\beta(\lambda)$ is now easily calculated.
Combining the scattering relation \eqref{MNbarrelation1} with the boundary conditions \eqref{MNbds} gives
\begin{equation}\label{beta_formula}
\beta(\lambda) = \lim_{x \to \infty} e^{-i \lambda x /\epsilon} \left( W_{+}(x;\lambda)-1 \right) \quad \text{for} \quad \lambda > 0.
\end{equation}
For a rational potential $u_0$ of the form \eqref{u0_def_frac}, we may explicitly write
\begin{equation}\label{base_reflect_coeff}
\beta(\lambda) = \frac{i}{\epsilon} e^{-2 \pi/\epsilon  \sum_{p=1}^{P} c_p} \int_{-\infty}^{\infty}  e^{-i h(z;\lambda)/\epsilon }  \left( u_0(z) - \sum_{p=1}^{P} \frac{v_p(\lambda)}{z-z_p}\right) \; dz,
\end{equation}
using  equation \eqref{MJost_eqn} for the Jost solution $W_{+}(x;\lambda)$ and Lemma~\ref{prop_f_to_infty}.
The coefficients $\{v_p(\lambda)\}_{p=1}^{P}$ are determined by the linear system \eqref{lin_sys_Abar}.

\subsection{Eigenfunctions and eigenvalues}\label{sec_eignfunc}

Eigenfunctions $w^+=\Phi(x;\lambda)$ corresponding to eigenvalues $\lambda <0$ are nontrivial solutions in $\mathbb{H}^+(\mathbb{R})$ of \eqref{w+bounded}, or equivalently \eqref{w+bounded-rational} for $u_0$ of the form \eqref{u0_def_frac}, with $w_0=0$.  Integrating \eqref{w+bounded-rational} with $w_0=0$ and imposing $\Phi(x;\lambda)\to 0$ as $x\to -\infty$ shows that $\Phi(x;\lambda)$ necessarily has the form
\begin{equation}\label{Phij_int}
\Phi(x;\lambda)
= - \frac{i}{\epsilon} e^{ih(x;\lambda)/\epsilon } \int_{-\infty}^{x} e^{-ih(z;\lambda)/\epsilon }  \sum_{p=1}^{P} \frac{\phi_p(\lambda)}{z-z_p} \; dz,
\end{equation}
where
$\phi_p(\lambda) := c_p \Phi(z_p;\lambda)$  for $p = 1,\ldots, P$.
Since $\lambda<0$, the convergent improper integral on the right-hand side of \eqref{Phij_int} becomes absolutely convergent if the contour is rotated at infinity into the upper-half $x$-plane.

While \eqref{Phij_int} satisfies \eqref{w+bounded-rational} and is analytic for $x$ in a neighborhood of $\mathbb{R}$, to obtain the required analyticity for all $x$ with $\I\{x\}>0$ requires conditions on both $\lambda<0$ and $\{\phi_p(\lambda)\}_{p=0}^P$.  As we will see these conditions will also guarantee that $\Phi(\cdot;\lambda)\in\mathbb{H}^+(\mathbb{R})$, i.e., that it vanishes as $|x|\to\infty$ in all directions of the upper half-plane.

\begin{definition}\label{defn_Uminus}
$U_m^<$ denotes any member of the equivalence class of contours in $\domain$ originating at $i\infty$ in $\domain_0$ and ending at $i\infty$ in $\domain_m$; see Figures~\ref{fig:fdomain} and \ref{fig:Upos_fig}.
\end{definition}

\begin{proposition}\label{lin_syst_eigfunc}
Let $\lambda<0$.  The function $\Phi(x;\lambda)$ is analytic in the upper-half $x$-plane if and only if  $\{\phi_p(\lambda)\}_{p=1}^P$ satisfy
\begin{equation}\label{Uzj_cond}
\int_{C_m^{<}} e^{-ih(z;\lambda)/\epsilon } \sum_{p=1}^{P} \frac{\phi_p(\lambda)}{z-z_p} \, dz = 0,\quad
 m = 1,  \dots, P, 
\end{equation}
where $C_m^{<} := U_m^{<}$ unless $ic_m/\epsilon$ is a strictly negative integer, in which case $C_m^{<} := \ell_0(z_m)$. (The contours $C_m^<$ extend to $i\infty$ in $\domain_0$ and $\domain_m$ making the integrals absolutely convergent.)
\end{proposition}

The proof of Proposition~\ref{lin_syst_eigfunc} is similar to that of Proposition~\ref{analyticity1_M} and Corollary~\ref{Um_neg_system} and is omitted for brevity.
Introducing the vector $\boldsymbol{\phi}(\lambda)=(\phi_1(\lambda),\dots,\phi_P(\lambda))^\mathsf{T}$ of unknowns, the conditions of Proposition~\ref{lin_syst_eigfunc} take the form of a square homogeneous linear system $\mathbf{A}^<(\lambda)\boldsymbol{\phi}(\lambda)=\mathbf{0}$ where the matrix $\mathbf{A}^<(\lambda)\in\mathbb{C}^{P\times P}$ has components
\begin{equation}
\label{lin_sys_coeff_minus}
A_{mp}^<(\lambda):= \int_{C_m^<}\frac{e^{-ih(z;\lambda)/\epsilon}}{z-z_p}\,dz, \quad \lambda<0.
\end{equation}
\begin{corollary}\label{cor:base_evans}
The function $\Phi(x;\lambda)$ given by \eqref{Phij_int} is both nontrivial and analytic in the closed upper-half $x$-plane if and only if 
$D(\lambda):=\det (\mathbf{A}^{<}(\lambda) ) =0$ and $\boldsymbol{\phi}(\lambda)$ is a nontrivial nullvector of $\mathbf{A}^<(\lambda)$, unique up to a constant multiple.
\end{corollary}

\begin{proof}
This follows from Proposition~\ref{lin_syst_eigfunc}. It only remains to check that $\mathrm{dim}\left(\mathrm{ker} \left( \mathbf{A}^{<}(\lambda)\right) \right) = 1$ whenever $\det(\mathbf{A}^<(\lambda))=0$, a fact shown in \cite{upcoming-paper}.
\end{proof}

\begin{corollary}\label{cor:phi_eigenvector}
Suppose $D(\lambda)=0$, and that the components $\{ \phi_p(\lambda) \}_{p=1}^{P}$ of the  nullvector $\boldsymbol{\phi}(\lambda)$   are normalized to satisfy
\begin{equation}\label{Phi_norm}
 \sum_{p=1}^{P} \phi_p(\lambda) = \lambda<0,
\end{equation}
then the function $\Phi(x;\lambda)$ satisfies the normalization condition \eqref{norm_eigfunc}, $\Phi(\cdot;\lambda) \in \mathbb{H}^{+}$, and therefore $\Phi(x;\lambda)$ is an eigenfunction of equation~\eqref{w+bounded}. 
\end{corollary}

\begin{proof}
By Corollary~\ref{cor:base_evans}, the function $\Phi(x;\lambda)$ is analytic in the closed upper-half $x$-plane. From \eqref{Phij_int} one then invokes the Residue Theorem to see that \eqref{Phi_norm} implies \eqref{norm_eigfunc}.  
 The fact that $\boldsymbol{\phi}(\lambda)$ can be scaled to satisfy \eqref{Phi_norm}, i.e., that $\ker(\mathbf{A}^<(\lambda))$ is not orthogonal to $(1,1,\dots,1)^\mathsf{T}$, is proven in \cite{upcoming-paper}.
Lastly, to show that $\Phi \in \mathbb{H}^{+}$, we check that $\Phi \to 0$ as $|x| \to \infty$ in the closed upper-half $x$-plane. 
First, applying Jordan's Lemma to \eqref{Phij_int} shows that $\Phi \to 0$ when $x \to \infty$ anywhere in $\domain_0$. To let $x\to\infty$ elsewhere in  $\domain$, we rewrite the limiting contour integral in \eqref{Phij_int} as
$\lim_{x\to \infty \in \domain_{m}} \int_{-\infty}^{x} = \int_{U_m^{<}}$ whenever $ic_m/\epsilon$ is not a negative integer.
From Proposition~\ref{lin_syst_eigfunc}, this limit is identically zero for all $m = 1,\ldots, P$, since each integral on $U_m^{<}$ vanishes when $D(\lambda)=0$.
\end{proof}
\begin{corollary}
\label{cor:normalize}
Let $D(\lambda)=0$ for some $\lambda=\lambda_j<0$.  The eigenfunction $\Phi_j(x):=\Phi(x;\lambda_j)$ given by \eqref{Phij_int} with the conditions from Corollary~\ref{cor:phi_eigenvector} satisfies the asymptotic condition \eqref{Phij_bds}. 
\end{corollary}
\begin{proof}
Applying l'H\^opital's rule to \eqref{Phij_int} and using $h'(x;\lambda)=\lambda+u_0(x)$, 
\begin{equation}
\lim_{x \to \infty}x \Phi(x;\lambda) 
=  \lim_{x \to \infty} \frac{ix^2}{ \epsilon + i\lambda x + ixu_0(x) } \sum_{p=1}^{P} \frac{\phi_p(\lambda)}{x-z_p}
=  \frac{1}{\lambda} \sum_{p=1}^{P} \phi_p(\lambda),
\end{equation}
and the desired result then follows from \eqref{Phi_norm}.
\end{proof}
The stronger result $\Phi(x;\lambda) = x^{-1} + O \left( x^{-2}\right)$ as $|x| \to \infty$ (uniformly in the closed upper-half $x$-plane)  is proven in \cite{upcoming-paper}.

Corollaries \ref{cor:base_evans}--\ref{cor:normalize} show that the equation $D(\lambda)=0$ involving the determinant of the $P\times P$ matrix $\mathbf{A}^<(\lambda)$ is exactly the condition that $\lambda=\lambda_j<0$ is an eigenvalue for the rational potential $u_0$ of the form \eqref{u0_def_frac}, and that $\Phi_j(x):=\Phi(x;\lambda_j)$ is the corresponding normalized eigenfunction.

\begin{remark}\label{rmk:integer_cm}
If the coefficient $i c_m/\epsilon$ is a positive integer $N_m$, the integral \eqref{lin_sys_coeff_minus} can be calculated by residues, and thus $A_{mp}^{<}(\lambda)=e^{-i \lambda z_m / \epsilon} \mathcal{P}_p(\lambda)$, where $\mathcal{P}_p(\lambda)$ is a polynomial in $\lambda$ of order at most $N_m$. If this holds for all $m=1,\dots,P$, then $D(\lambda)$ is proportional via a nonvanishing exponential factor to a polynomial of order at most 
$\prod_{m=1}^{P} N_m$.  If $P=1$, this polynomial is a scaled Laguerre polynomial; see \cite{KodamaAblowitzSatsuma82}.
\end{remark}

\subsection{Analytic continuation, Evans function, and phase constants}\label{sec_analytic}

The matrix $\mathbf{A}^<(\lambda)$ defined for $\lambda<0$ by \eqref{lin_sys_coeff_minus} has an analytic continuation to the maximal domain $\mathbb{C}\setminus\mathbb{R}^+$.  Recalling that $h(z;\lambda)$ is dominated for large $|z|$ by the term $\lambda z$, it is clear that this continuation is afforded simply by rotating the infinite ``tails'' of the integration contours $C_m^<$ so that they tend to complex $\infty$ in the direction $\arg(z)=- \arg(i\lambda)$.  In the process of rotating the contours from their initially upward vertical configuration when $\lambda<0$, the function $f(z)$ appearing in the integrand via $h(z;\lambda)$ must be analytically continued through its vertical branch cuts as well.  
Therefore we observe that $D(\lambda):=\det(\mathbf{A}^<(\lambda))$ is an \emph{Evans function} for the spectral problem with rational potential $u_0$ of the form \eqref{u0_def_frac}; it is an analytic function in the domain $\mathbb{C}\setminus\mathbb{R}^+$ (complementary to the continuous spectrum for the problem) whose roots are precisely the eigenvalues.  

If the analytic continuation of $\mathbf{A}^<(\lambda)$ is carried out through the upper-half $\lambda$-plane to the cut $\mathbb{R}^+$, the contour $C^<_m$ will be rotated through the left-half $z$-plane to coincide precisely with $C^>_m$; compare the two panels of Figure~\ref{fig:Upos_fig}.  Thus, the analytic continuation of the matrix $\mathbf{A}^<(\lambda)$ from $\lambda<0$ through the upper half-plane to $\lambda>0$ coincides with the matrix $\mathbf{A}^>(\lambda)$ defined in \eqref{lin_sys_coeff}.
In a similar way, the 
vector $\mathbf{b}^{<}(\lambda)$ defined for $\lambda<0$ with components
\begin{equation}
b_m^{<}(\lambda) := -\lambda \int_{C_m^<} e^{-ih(z;\lambda)/\epsilon } \, dz, \quad \lambda<0,
\label{b-coeff-minus}
\end{equation}
has an analytic continuation to $\mathbb{C}\setminus\mathbb{R}^+$, taking a boundary value on $\mathbb{R}^+$ from the upper half-plane that coincides with the vector $\mathbf{b}^>(\lambda)$ also defined in \eqref{lin_sys_coeff}.  
This shows that indeed for each $x$ with $\I\{x\}\ge 0$, the Jost solution $W_+(x;\lambda)$ given by the formula \eqref{MJost_eqn} is the boundary value on $\mathbb{R}^+$ from $\I\{\lambda\}>0$ of a function $W(x;\lambda)$ analytic for $\lambda\in\mathbb{C}\setminus\mathbb{R}^+$ with the possible exception of the eigenvalues $\lambda<0$ satisfying $D(\lambda)=0$; only at these points do the quantities $\{v_p(\lambda)\}_{p=1}^P$ entering into \eqref{MJost_eqn} become indeterminate as the matrix of the analytic continuation of the system \eqref{lin_sys_Abar} through $0\le\arg(\lambda)<2\pi$ becomes singular.
This idea allows us to deduce the remaining scattering data corresponding to $u_0$ of the form \eqref{u0_def_frac}, namely the phase constants $\{\gamma_j\}_{j=1}^N$ corresponding to the negative eigenvalues $\{\lambda_j\}_{j=1}^N$.  The first step is to determine how analytic continuation of the system \eqref{lin_sys_Abar} fails near an eigenvalue $\lambda=\lambda_j$.
\begin{proposition}\label{vk_Lau_exp}
Let $\mathbf{v}(\lambda)$ be the unique solution of $\mathbf{A}^<(\lambda)\mathbf{v}(\lambda)=\mathbf{b}^<(\lambda)$ for each $\lambda\in\mathbb{C}\setminus\mathbb{R}^+$ for which $D(\lambda)\neq 0$.  Then each component $v_p(\lambda)$ is analytic in $\mathbb{C}\setminus\mathbb{R}^+$ except at the eigenvalues $\{\lambda_j\}_{j=1}^N$ which are simple poles, with corresponding
Laurent expansion
\begin{equation}\label{v_Laurent_expansion}
v_p(\lambda)= - i \epsilon \frac{ \phi_p(\lambda_j)}{\lambda-\lambda_j} +  (z_p + \Gamma_j) \phi_p(\lambda_j) + O\left( \lambda-\lambda_j \right)
\quad
\text{as}
\quad
\lambda\to\lambda_j
\end{equation}
where $\Gamma_j$ is a constant (independent of $p$) given by
\begin{equation}\label{base_phase_constant}
\Gamma_j := -\frac{i \epsilon}{2 \lambda_j} -\frac{1}{2 \lambda_j} \sum_{p=1}^{P} z_p \phi_p(\lambda_j) -\frac{i \epsilon}{2} \frac{\mathbf{m}^\mathsf{T}{\mathbf{b}^{<}}'(\lambda_j)}{\mathbf{m}^\mathsf{T} \mathbf{b}^{<}(\lambda_j)},
\end{equation}
and where $\{\phi_p(\lambda_j)\}_{p=1}^{P}$ are the components of the right nullvector of $\mathbf{A}^{<}(\lambda_j)$ normalized by \eqref{Phi_norm} and $\mathbf{m}$ is a nonzero left nullvector of $\mathbf{A}^{<}(\lambda_j)$. 
\end{proposition}

We include the proof of Proposition~\ref{vk_Lau_exp} in the Appendix.  The following result then shows that the phase constant associated with each eigenvalue $\lambda_j<0$ by means of the Laurent expansion \eqref{LaurentW} is given by $\gamma_j=\Gamma_j$.

\begin{corollary}\label{W_Laurent_exp}
The function $W(x;\lambda)$ whose boundary value from $\I\{\lambda\}>0$ on $\mathbb{R}^+$ is the Jost solution $W_{+}(x;\lambda)$ is meromorphic on $\mathbb{C}\setminus\mathbb{R}^+$ with simple poles only at the eigenvalues $\{\lambda_j\}_{j=1}^N$, at each of which it has a Laurent expansion of the form \eqref{LaurentW} with $\Phi_j(x):=\Phi(x;\lambda_j)$ defined as in Section~\ref{sec_eignfunc} and with $\gamma_j=\Gamma_j$ defined as in \eqref{base_phase_constant}.
\end{corollary}

\begin{proof}
The function $W(x;\lambda)$ is given by \eqref{MJost_eqn} with the coefficients $\{v_p(\lambda)\}_{p=1}^P$ obtained by solving $\mathbf{A}^<(\lambda)\mathbf{v}(\lambda)=\mathbf{b}^<(\lambda)$, and with the contour of integration rotated appropriately to ensure absolute convergence of the integral. Substitution of \eqref{v_Laurent_expansion} into \eqref{MJost_eqn} and using \eqref{Phij_int} yields, after some manipulation, \eqref{LaurentW} in which $\gamma_j=\Gamma_j$.
\end{proof}

% Acknowledgements
\section*{Acknowledgments}
The authors gratefully acknowledge support from the National Science Foundation (NSF) under grant DMS-1206131. 

\begin{appendix}
\section*{Appendix: Proofs}\label{appdx_proofs}

%%%%%%%%%%%%
%\subsection*{Proof of Proposition~\ref{analyticity1_M}}\label{Jost_Appx}
%%%%%%%%%%%%

\begin{proof}[Proof of Proposition~\ref{analyticity1_M}:]
Fix $\lambda>0$.  The conditions of the proposition are equivalent to $W_+(x)=W_{+}(x;\lambda)$ being analytic at each possible singularity $z_1,\dots,z_P$ in the upper-half $x$-plane. 
For any $m=1,\dots,P$, we use \eqref{def_Hj} defining $H_m(x)=H_m(x;\lambda)$ in \eqref{MJost_eqn} to express
$W_+(x)$ in the form
\begin{equation}
W_+(x)=g_m(x)(i(x-z_m))^{ic_m/\epsilon}\int_{\ell_0(x)}
(i(z-z_m))^{-ic_m/\epsilon-1}H_m(z)\,dz,
\label{eq:W-plus-minus-one-rewrite}
\end{equation}
where we have also introduced a function $g_m(x)$ analytic and nonvanishing at $x=z_m$
defined precisely by the identity $\epsilon g_m(x)(i(x-z_m))^{ic_m/\epsilon}=ie^{ih(x)/\epsilon}$.
We will now use Cauchy's Theorem to deform the path $\ell_0(x)$ of integration so as to pass directly through the branch point $z_m$ on the way to a nearby point $x\in\domain$.  While $H_m(z)$ is analytic at $z=z_m$, the other factor in the integrand of \eqref{eq:W-plus-minus-one-rewrite} may not be integrable at $z=z_m$.

First suppose that $\R\{ i c_m/\epsilon\} < 0$.  The factor $(i(z-z_m))^{-ic_m/\epsilon-1}$ is integrable at $z=z_m$ so we can apply Cauchy's Theorem to 
obtain
\begin{multline}
W_{+}(x) = g_m(x)(i(x-z_m))^{ ic_m/\epsilon} \int_{\ell_0(z_m)}(i(z-z_m))^{- ic_m/\epsilon-1}H_m(z)\,dz \\ {}+g_m(x)(i(x-z_m))^{ic_m/\epsilon}\int_{z_m}^x
(i(z-z_m))^{-ic_m/\epsilon-1}H_m(z)\,dz.
\label{eq:W-plus-deform-good-case}
\end{multline}
The second term on the right-hand side is analytic for $x$ near $z_m$.  Indeed, substituting for $H_m(z)$ its power series in $i(z-z_m)$, 
\begin{equation}
H_m(z)=\sum_{k=0}^\infty \eta_{mk}(i(z-z_m))^k,
\label{eq:Hm-series}
\end{equation}
for small enough $|z-z_m|$, term-by-term integration yields
\begin{multline}\label{gm_term-by-term}
g_m(x)(i(x-z_m))^{ic_m/\epsilon}\int_{z_m}^x(i(z-z_m))^{-ic_m/\epsilon-1}H_m(z)\,dz\\
{}=-i g_m(x)\sum_{k=0}^\infty\frac{\eta_{mk}}{k-ic_m/\epsilon}(i(x-z_m))^k.
\end{multline}
Now, \eqref{gm_term-by-term} is analytic at $x=z_m$ as it is a product of an analytic function and a convergent power series.  So, requiring that $W_{+}(x)$ be analytic at $x=z_m$ is equivalent to demanding that the first term on the right-hand side of \eqref{eq:W-plus-deform-good-case} be analytic at $x=z_m$.  But this term is the product of (i) $g_m(x)$ nonvanishing at $x=z_m$, (ii) $(i(x-z_m))^{ic_m/\epsilon}$ which blows up as $x\to z_m$  since $\R\{ic_m/\epsilon\}<0$, and (iii) the $x$-independent integral on $\ell_0(z_m)$. Hence, the analyticity of $W_+(x)$ at $x=z_m$ is equivalent to the vanishing of this integral (cf.\@ \eqref{lzj_cond}).

Next let $\R\{ ic_m/\epsilon\} \geq 0$ with $\omega_m\neq 0$.  Since $(i(z-z_m))^{-ic_m/\epsilon-1}$ is branched and nonintegrable at $z=z_m$,  we integrate by parts $\mu_m+1$ times (note $\mu_m\geq 0$) to restore integrability before deforming the contour $\ell_0(x)$.  With $K_m:=i^{\mu_m+1}\Gamma(\omega_m)/\Gamma(\omega_m+\mu_m+1)\neq 0$ we have
\begin{equation}\label{ddz(z-zm)}
(i(z-z_m))^{-ic_m/\epsilon-1}=K_m\frac{d^{\mu_m+1}}{dz^{\mu_m+1}}(i(z-z_m))^{-\omega_m}.
\end{equation}
Integrating \eqref{eq:W-plus-minus-one-rewrite} by parts using \eqref{ddz(z-zm)}, note that $(i(z-z_m))^{-\omega_m}$ is integrable at $z=z_m$ ($0\leq \R\{\omega_m\}<1$), so Cauchy's Theorem gives
\begin{multline}
W_+(x)= (-1)^{\mu_m+1}K_mg_m(x)(i(x-z_m))^{ic_m/\epsilon}\int_{\ell_0(z_m)}\frac{H_m^{(\mu_m+1)}(z)}{(i(z-z_m))^{\omega_m}}\,dz\\
{} +(-1)^{\mu_m+1}K_mg_m(x)(i(x-z_m))^{ic_m/\epsilon}\int_{z_m}^x
\frac{H_m^{(\mu_m+1)}(z)}{(i(z-z_m))^{\omega_m}}\,dz\\
{} +K_m g_m(x)(i(x-z_m))^{ic_m/\epsilon}\sum_{k=0}^{\mu_m} (-1)^k H_m^{(k)}(x)\frac{d^{\mu_m-k}}{dx^{\mu_m-k}}\frac{1}{(i(x-z_m))^{\omega_m}}.
\label{eq:three-lines}
\end{multline}
Now, recalling $ic_m/\epsilon=\mu_m+\omega_m$, note that the terms on the third line of \eqref{eq:three-lines} are analytic at $x=z_m$.  Likewise, with the help of the uniformly convergent power series \eqref{eq:Hm-series} and term-by-term differentiation and integration (assuming only that $|x-z_m|$ is sufficiently small), the second line of \eqref{eq:three-lines} is analytic at $x=z_m$.  The first line of the right-hand side of \eqref{eq:three-lines} is the product of (i) $K_mg_m(x)$ analytic and nonzero at $x=z_m$, (ii) a branched factor $(i(x-z_m))^{ic_m/\epsilon}$, and (iii) the definite integral on $\ell_0(z_m)$. This term of $W_+(x)$ is therefore analytic if and only if the integral on $\ell_0(z_m)$ vanishes (cf.\@ \eqref{lzj_der_cond}).  

If $\R\{ic_m/\epsilon\}\ge 0$ but instead $\omega_m=0$, then \eqref{eq:W-plus-minus-one-rewrite} 
becomes
\begin{equation}
W_+(x)=-ig_m(x)(x-z_m)^{\mu_m}\int_{\ell_0(x)}\frac{H_m(z)}{(z-z_m)^{\mu_m+1}}\,dz,\;\; \mu_m=0,1,2,3,\dots.
\end{equation}
Since $g_m$ and $H_m$ are analytic at $z_m$, $W_+(x)$ will be analytic at $z_m$ exactly when the integral factor is single-valued as a function of $x$ (cf.\@ \eqref{lzj_Log_cond}).
\end{proof}

%%%%%%%%%%%%
%\subsection*{Proof of Proposition~\ref{vk_Lau_exp}}\label{Phs_const_Appx}
%%%%%%%%%%%%

\begin{proof}[Proof of Proposition~\ref{vk_Lau_exp}:]
Since $\mathbf{A}(\lambda)=\mathbf{A}^<(\lambda)$ given by \eqref{lin_sys_coeff} and 
$\mathbf{b}(\lambda)=\mathbf{b}^<(\lambda)$ given by
\eqref{b-coeff-minus} are analytic for $\lambda\in\mathbb{C}\setminus\mathbb{R}^+$, the solution 
$\mathbf{v}(\lambda)$ of the system $\mathbf{A}(\lambda)\mathbf{v}(\lambda)=\mathbf{b}(\lambda)$ will be analytic in the same domain with the possible exception of the eigenvalues at which $D(\lambda)=0$.  In \cite{upcoming-paper} it is shown that these isolated singular points are simple poles of $\mathbf{v}(\lambda)$, which therefore has a Laurent expansion about $\lambda=\lambda_j$ of the general form
\begin{equation}
\mathbf{v}(\lambda)=\frac{\mathbf{v}_j^{[-1]}}{\lambda-\lambda_j} +\mathbf{v}_j^{[0]} + (\lambda-\lambda_j)\mathbf{v}_j^{[1]} + O\left((\lambda-\lambda_j)^2\right)
\quad
\text{as}
\quad
\lambda \to \lambda_j.
\end{equation}
Substitution into $\mathbf{A}(\lambda)\mathbf{v}(\lambda)=\mathbf{b}(\lambda)$ implies:
\begin{align}
\mathbf{A}(\lambda_j)\mathbf{v}_j^{[-1]}&=\mathbf{0}, 
\label{O-1}\\
\mathbf{A}(\lambda_j)\mathbf{v}_j^{[0]}&=\mathbf{b}(\lambda_j)-\mathbf{A}'(\lambda_j)\mathbf{v}^{[-1]}_j,  
\label{O0}\\
 \mathbf{A}(\lambda_j)\mathbf{v}_j^{[1]}&=\mathbf{b}'(\lambda_j)-\frac{1}{2}\mathbf{A}''(\lambda_j)\mathbf{v}^{[-1]}_j -\mathbf{A}'(\lambda_j)\mathbf{v}_j^{[0]}. 
 \label{O1}
\end{align}
Since $\mathbf{A}(\lambda_j)$ is singular and has maximal rank $P-1$ (see \cite{upcoming-paper}), the first two equations \eqref{O-1} and \eqref{O0} have the general solutions
\begin{equation}\label{vj-1_c0}
\mathbf{v}_j^{[-1]}=\alpha_j^{[-1]} \boldsymbol{\phi}(\lambda_j)\quad\text{and}\quad
\mathbf{v}_j^{[0]}=\mathbf{p}+\alpha_j^{[0]}\boldsymbol{\phi}(\lambda_j),
\end{equation}
where $\mathbf{p}$ is a particular solution of \eqref{O0} and $\boldsymbol{\phi}(\lambda_j)$ is the nullvector of $\mathbf{A}(\lambda_j)$ normalized according to \eqref{Phi_norm}.  Here $\alpha_j^{[-1]}$ and $\alpha_j^{[0]}$ are constants determined by the fact that the right-hand sides of the equations
\eqref{O0} and \eqref{O1} lie in the range of the singular matrix $\mathbf{A}(\lambda_j)$.  Therefore, multiplying 
these latter two equations on the left by $\mathbf{m}^\mathsf{T}$ where $\mathbf{m}$ is any nonzero left nullvector of $\mathbf{A}(\lambda_j)$ (unique up to scaling) and using \eqref{vj-1_c0} gives  
\begin{equation}\label{c0_choice}
\alpha_j^{[-1]}=\frac{\mathbf{m}^\mathsf{T}\mathbf{b}(\lambda_j)}{\mathbf{m}^\mathsf{T}\mathbf{A}'(\lambda_j)
\boldsymbol{\phi}(\lambda_j)}
\end{equation}
and
\begin{equation}\label{c1_choice}
\alpha_j^{[0]}  = \frac{\mathbf{m}^\mathsf{T}\left( \mathbf{b}'(\lambda_j)- \tfrac{1}{2}\alpha_j^{[-1]} \mathbf{A}''(\lambda_j) \boldsymbol{\phi}(\lambda_j) -\mathbf{A}'(\lambda_j) \mathbf{p}\right)}{\mathbf{m}^{\mathsf{T}} \mathbf{A}'(\lambda_j)\boldsymbol{\phi}(\lambda_j)}.
\end{equation}
To simplify these formulas,
we differentiate \eqref{lin_sys_coeff_minus} to obtain the identity
\begin{equation}\label{Amatplj_eqn}
\mathbf{A}'(\lambda) = \frac{i}{\epsilon \lambda} \mathbf{B}(\lambda) - \frac{i}{\epsilon} \mathbf{A}(\lambda) \mathbf{Z},
\end{equation}
where $\mathbf{B}(\lambda) := \left[ \mathbf{b}(\lambda)  \cdots  \mathbf{b}(\lambda) \right]$ and $\mathbf{Z} := \diag(z_1, \ldots, z_N)$. Since $\mathbf{m}^\mathsf{T}\mathbf{A}(\lambda_j)=\mathbf{0}^{\mathsf{T}}$, 
\eqref{c0_choice} becomes
\begin{equation}\label{c0_simple}
\alpha_j^{[-1]}
= -i\epsilon\frac{\lambda_j\mathbf{m}^\mathsf{T}\mathbf{b}(\lambda_j)}{\mathbf{m}^\mathsf{T} \mathbf{B}(\lambda_j) \boldsymbol{\phi}(\lambda_j)}
= -i\epsilon\frac{\mathbf{m}^\mathsf{T}\mathbf{b}(\lambda_j)}{\mathbf{m}^\mathsf{T} \mathbf{b}(\lambda_j)}
= -i \epsilon,
\end{equation}
due to \eqref{Phi_norm}. (That $\mathbf{m}^\mathsf{T} \mathbf{b}(\lambda_j) \neq 0$ is proven in \cite{upcoming-paper}.)
Now, substituting \eqref{vj-1_c0}, \eqref{Amatplj_eqn}, and \eqref{c0_simple} into \eqref{O0} and using \eqref{Phi_norm} again yields
\begin{equation}
\mathbf{A}(\lambda_j)\mathbf{p} = \mathbf{A}(\lambda_j) \mathbf{Z} \boldsymbol{\phi}(\lambda_j),
\end{equation}
and hence the particular solution $\mathbf{p}$ may be taken as $\mathbf{p} = \mathbf{Z} \boldsymbol{\phi}(\lambda_j)$.
Finally, substituting \eqref{Amatplj_eqn} and its derivative into \eqref{c1_choice} shows that
$\alpha_j^{[0]}=\Gamma_j$ as defined by \eqref{base_phase_constant}.
\end{proof}
\end{appendix}

\bibliographystyle{sapm}% BST file

\begin{thebibliography}{10}
\providecommand{\url}[1]{\texttt{#1}}
\providecommand{\urlprefix}{URL }
\expandafter\ifx\csname urlstyle\endcsname\relax
  \providecommand{\doi}[1]{doi:\discretionary{}{}{}#1}\else
  \providecommand{\doi}{doi:\discretionary{}{}{}\begingroup
  \urlstyle{rm}\Url}\fi
\providecommand{\eprint}[2][]{\url{#2}}

\bibitem{Benjamin67}
\textsc{T.~B. Benjamin}, Internal waves of permanent form in fluids of great
  depth, \emph{J. Fluid Mech.} 29:559--592 (1967).

\bibitem{DavisAcrivos67}
\textsc{R.~E. Davis} and \textsc{A.~Acrivos}, Solitary internal waves in deep
  water, \emph{J. Fluid Mech.} 29:593--607 (1967).

\bibitem{Ono75}
\textsc{H.~Ono}, Algebraic solitary waves in stratified fluids, \emph{J. Phys.
  Soc. Jpn.} 39:1082--1091 (1975).

\bibitem{KodamaAblowitzSatsuma82}
\textsc{Y.~Kodama}, \textsc{M.~J. Ablowitz}, and \textsc{J.~Satsuma}, Direct
  and inverse scattering problems of the nonlinear intermediate long wave
  equation, \emph{J. Math. Phys.} 23:564--576 (1982).

\bibitem{Xuthesis}
\textsc{Z.~Xu}, \emph{Asymptotic analysis and numerical analysis of the
  {B}enjamin-{O}no equation}, Ph.D. thesis The University of Michigan (2010).

\bibitem{upcoming-paper2}
\textsc{P.~D. Miller} and \textsc{A.~N. Wetzel}, The scattering transform for
  the {B}enjamin-{O}no equation in the small-dispersion limit (2015), In
  preparation.

\bibitem{BockKruskal79}
\textsc{T.~L. Bock} and \textsc{M.~D. Kruskal}, A two-parameter {M}iura
  transformation of the {B}enjamin-{O}no equation, \emph{Phys. Lett. A}
  74:173--176 (1979).

\bibitem{Nakamura79b}
\textsc{A.~Nakamura}, B{\"a}cklund transform and conservation laws of the
  {B}enjamin-{O}no equation, \emph{J. Phys. Soc. Jpn.} 47:1335--1340 (1979).

\bibitem{FokasAblowitz83}
\textsc{A.~S. Fokas} and \textsc{M.~J. Ablowitz}, The inverse scattering
  transform for the {B}enjamin-{O}no equation: A pivot to multidimensional
  problems, \emph{Stud. Appl. Math.} 68:1--10 (1983).

\bibitem{KaupMatsuno98}
\textsc{D.~J. Kaup} and \textsc{Y.~Matsuno}, The inverse scattering transform
  for the {B}enjamin-{O}no equation, \emph{Stud. Appl. Math.} 101:73--98
  (1998).

\bibitem{Gakhov66}
\textsc{F.~D. Gakhov}, \emph{Boundary Value Problems}, Dover Publications,
  1990.

\bibitem{Muskhelishvili53}
\textsc{N.~I. Muskhelishvili}, \emph{Singular Integral Equations}, Dover
  Publications, 2008.

\bibitem{ReedSimon4}
\textsc{M.~Reed} and \textsc{B.~Simon}, \emph{Methods of Modern Mathematical
  Physics IV: Analysis of Operators}, Academic Press, 1978.

\bibitem{AblowitzFokasAnderson1983}
\textsc{M.~J. Ablowitz}, \textsc{A.~S. Fokas}, and \textsc{R.~L. Anderson}, The
  direct linearizing transform and the {B}enjamin-{O}no equation, \emph{Phys.
  Lett. A} 93:375--378 (1983).

\bibitem{SantiniAblowitzFokas1984}
\textsc{P.~Santini}, \textsc{M.~J. Ablowitz}, and \textsc{A.~S. Fokas}, On the
  limit from the intermediate long-wave equation to the {B}enjamin-{O}no
  equation, \emph{J. Math. Phys.} 25:892--899 (1984).

\bibitem{AndersonTaflin1985}
\textsc{R.~L. Anderson} and \textsc{E.~Taflin}, The {B}enjamin-{O}no equation
  --- {R}ecursivity of linearization maps --- {L}ax pairs, \emph{Lett. Math.
  Phys.} 9:299--311 (1985).

\bibitem{upcoming-paper}
\textsc{A.~N. Wetzel}, \emph{Three stratified fluid models: quasi-geostrophy,
  {B}enjamin-{O}no, and tidal resonance}, Ph.D. thesis The University of
  Michigan (2015), In preparation.

\end{thebibliography}

\end{document}